%% file: main.tex
\documentclass[11pt,letterpaper]{article}
\usepackage{fullpage}
\newif\ifanon
\anonfalse
\newif\ifconf
\conffalse
\usepackage[draft]{ktmacros}
\usepackage{algpseudocode}

\usepackage{etoolbox} 

\usepackage{natbib}

\newcommand{\agg}{\ensuremath{\mathcal{A}}}

\newcommand{\countnonzero}{{\sc CountNonZero}}
\newcommand{\Out}{\ensuremath{\textbf{Out}}}
\newcommand{\OS}{\ensuremath{\mathcal{O}}}
\newcommand{\x}{\ensuremath{\mathbf{x}}}
\newcommand{\y}{\ensuremath{\mathbf{y}}}

\title{Local Pan-Privacy for Federated Analytics}
\ifanon
\author{Anonymous Submission\\}
\else
\author{
Vitaly Feldman\footnote{Apple.} 
\and 
Audra McMillan\footnotemark[1]
\and 
Guy N. Rothblum\footnotemark[1] 
\and 
Kunal Talwar\footnotemark[1]
}
\fi

\usepackage{amsmath}
\newif\ificml
\icmlfalse

\begin{document}
\maketitle
\begin{abstract}
\input{Sections/abstract}

\end{abstract}

\input{Sections/content}

\bibliography{refs}
\bibliographystyle{plainnat}
\appendix
\input{Sections/appendix}

\end{document}

%% file: Sections/abstract.tex
Pan-privacy was proposed by~\cite{DworkNPRY10} as an approach to designing a private analytics system that retains its privacy properties in the face of intrusions that expose the system's internal state. Motivated by federated telemetry applications, we study {\em local pan-privacy}, where privacy should be retained under repeated unannounced intrusions {\em on the local state}. We consider the problem of monitoring the count of an event in a federated system, where event occurrences on a local device should be hidden even from an intruder on that device. We show that under reasonable constraints, the goal of providing information-theoretic differential privacy under intrusion is incompatible with collecting telemetry information. We then show that this problem can be solved in a scalable way using standard cryptographic primitives.

%% file: Sections/content.tex
\section{Introduction}

Private federated telemetry systems allow 
for collection of aggregate statistics from a population, while ensuring a strong privacy guarantee for individuals. 
For example, private federated learning and statistics have been used to collect webpage and search engine popularities from Chrome browsers~\citep{ErlingssonPK14}, learn popular emojis~\citep{Apple2017}, collect Covid epedimiological metrics~\citep{ENPA:2021}, collect browser performance metrics~\citep{firefox}, collect operating system telemetry~\citep{DingKY17}, and train keyboard models~\citep{XuZACKMRZ23, ZhangRXZZK23}.

These systems typically rely on the client device storing information about usage on device, and periodically, at appropriate times, taking part in a protocol that computes an aggregate. These protocols use encryption to protect this data from an eavesdropper. In some cases, additional cryptographic protocols are used to protect the individual contributions from the server computing the aggregate~\citep{Bonawitz17, Corrigan-GibbsB17}. Finally, the aggregate itself protects individual contributions by noise addition to provide a differential privacy guarantee.

In some applications, one may want to additionally protect against an attacker that has access to the device. For example, a computer in a public library may be accessed by one user, and later by another. In such a case, we would want to ensure that the second user cannot learn about the activity of the first user by inspecting the internal state of the device. In this work, we initiate the study of estimating population statistics while ensuring privacy against an on-device intruder. We call this model {\em local pan-privacy}, as it aims to  protect against an intrusion at the local device, much as pan-privacy protects against an intrusion at the server.

We study locally pan-private algorithms some fundamental statistical tasks in a simple streaming setting. There are $n$ devices in total. For each device, at each time step, an event of interest, such as an application crash while using a feature, may or may not occur. Thus, each device's input is a sequence of $T$ bits. The device can communicate with the server at the end of the $T$ time steps, and we require that the protocol retains its privacy properties in the face of multiple intrusions on device (c.f. \cref{remark:definition}). The first statistic we study is {\em \countnonzero}, which counts the number of devices on which the event occurred in at least one of the $T$ time steps. In the standard local privacy model, this task can be accomplished by each device sending a randomized response of a bit corresponding to whether on not the event occurred at any of the $T$ steps. Such a protocol can also lead to near-optimal central privacy guarantees via shuffling or aggregation~\cite{FeldmanMT:2020}, when the responses are aggregated by a trusted server, or by a secure aggregation system. In this work, we focus on the setting where we either have a trusted server, or where the secure aggregation system uses a two-server architecture as in PRIO~\cite{Corrigan-GibbsB17}.

We also study two additional statistics of the event count distribution: the {\em mean} number of occurrences per device, and a {\em histogram} of the number of occurrences, appropriately bucketed. As with the \countnonzero, these tasks also admit simple and commonly-used algorithms in the local privacy model, based on the Laplace mechanism\citep{Dwork:2006} and on {\sc Rappor}\citep{ErlingssonPK14} respectively. We remark that these simple tasks underpin a large number of private federated statistics, and can enable a much richer set of data science tasks. For example, ~\citet{zhu2020federated, ChadhaCDFHJMT24} use histograms over small known domains to discover heavy hitters over large domains.

We show that local pan-privacy is severely limited if we insist on providing (information-theoretic) differential privacy in the face of intrusions. Our lower bound shows that for the \countnonzero\ task, the error of any algorithm must be $\sqrt{T}$ times larger than that needed for local differential privacy alone. This lower bound holds even when the event occurs at most once on each device.

\begin{theorem}[Informal version of~\cref{thm:main_lb}]
    Any locally pan-private algorithm (for $\eps=1$) for \countnonzero\ on $n$ devices, for large enough $T$, must incur additive error $\Omega(\sqrt{nT})$, even though a local DP algorithm can estimate \countnonzero\ with additive error $\Omega(\sqrt{n})$.
\end{theorem}

We then show that under standard cryptographic assumptions, local pan-privacy can be ensured without this overhead. We present algorithms for all of the aforementioned problems, for both the single- and the two-server models, showing that local pan-privacy comes at {\em no} additional cost in the privacy-utility trade-off. Our protocols need a public-key encryption scheme, with a few additional properties that are satisfied by commonly-used encryption schemes. We note that public-key encryption is already used to protect the communication from a network intrusion, so local pan-privacy does not increase the complexity of the required  assumptions. While we can define local pan-privacy broadly as a computational differential privacy guarantee against a local intruder, our protocols actually provide a stronger semantic security guarantee.

\begin{theorem} [Informal version of ~\cref{thm:ub_countnonzero}]
\label{thm:ub_informal}
   Suppose that we have a public-key encryption scheme that is rerandomizable. Then there is a streaming algorithm for \countnonzero\ in the single-server and in the two-server model with the following properties.
    \begin{itemize}
        \item The on-device algorithm satisfies computational local pan-privacy; the on-device sequence of states on any pair of input streams are computationally indistinguishable.
        \item The on-device state consists of $O(1)$ ciphertexts, and the device sends one message consisting of $O(1)$ ciphertexts.
        \item In the local and aggregator model of differential privacy, the algorithm achieves privacy-utility trade-offs that are within constant factors of algorithms without the local pan-privacy constraint.
    \end{itemize}
\end{theorem}

In a rerandomizable encryption scheme, a ciphertext can be ``rerandomized'' (using only the public key) to create a new ciphertext encrypting the same plaintext, which is indistinguishable from a fresh encryption of a different plaintext, see \cref{def:rerandomizable}. 
A similar result to \cref{thm:ub_informal} holds for building approximate histograms, as well as for estimating the mean number of events. The latter requires slightly stronger assumptions on the encryption scheme (namely, the scheme needs to be additively homomorphic).

At a high level, this is achieved by maintaining a state on device that contains information about the stream so far, but in encrypted form. Since the local device does not have the private key, this ensures that the on-device state contains no useful information for an adversary that does not collude with the server. The challenge then is to maintain the state under updates to the stream, and we show that for our tasks of interest, this is feasible. We remark that a public-key fully homomorphic encryption (FHE) scheme can be used to achieve computational local pan-privacy, as any algorithm for the on-device computation can be made locally pan-private by keeping the state encrypted at all times and operating on the encrypted state. However, this is overkill: FHE schemes incur a large space and time overhead, and in this work we strive to rely on more minimal and more efficient cryptographic primitives.

 We provide privacy at the level of the whole device, and not just at the event level. Thus a user may use the shared device multiple times during our collection period, and our algorithms will protect all of their interactions. Our model allows for continuous intrusion: the adversary can see the state of the device before and after each update. In our example of a shared device, this means that the attacker can see the memory contents of the device before an after a user is using it (but not while they are using it). Thus from our point of view, the updates to the state are atomic.
Finally, we show that the existence of a public-key encryption scheme is necessary for the existence of an accurate locally pan-private algorithm for \countnonzero.
\begin{theorem}[Informal version of ~\cref{thm:pkeneeded}]
Suppose that we have a locally pan-private algorithm for \countnonzero\ that has additive error less than $n/4$ with high probability. Then we can build a public key encryption scheme\end{theorem}

\subsection{Related Work}

Pan-privacy~\citep{DworkNPRY10} was proposed as an approach to designing data analysis algorithms that do not maintain a disclosive internal state. In particular, they maintain privacy even under intrusions. This notion was studied in several subsequent works, e.g. \citep{MirMNW11}.
The goal of maintaining privacy in the presence of potential intrusions on a user's device is common to several settings. Examples include web browsing, where several popular browsers offer a mode where browsing history is not stored on device and chat apps where messages are ephemeral. In many situations, it may be undesirable to keep any information from events in such settings. In others, statistical information about events such as application crashes may be useful to improve the user experience. The notion of history independence~\citep{Micciancio97,NaorT01} shares similar goals of ensuring that the memory representation of a data structure does not leak information about the history of interactions with a data structure. Oblivious RAM algorithms~\citep{GoldreichO96} share a similar goal of ensuring that the memory access pattern of an algorithm does not leak information about the inputs it is being run on. Forward Secrecy~\citep{Gunther90} addresses a similar goal of protecting old communications from future intrusions.

\subsection*{Organization:} We present some preliminaries and definitions in \cref{sec:prelims}. \cref{sec:lb} presents our lower bound for the information-theoretic privacy case. We present algorithms in the single- and two-server settings in \cref{sec:ub-single} and \cref{sec:ub-two-server} respectively. \cref{sec:pkeneeded} shows that public-key cryptography is needed, and we conclude with some open problems in \cref{sec:conclusions}.

\section{Definitions and Preliminaries}
\label{sec:prelims}
We first recall the definition of differential privacy.
\begin{definition}[ $(\eps,\delta)$-Indistinguishability] We say that two random variables $Y$ and $Y'$ on a finite set $R$ are $(\eps,\delta)$ indistinguishable if for for every $S \subseteq R$,
  \begin{align*}
    \Pr[Y \in S] &\leq e^{\eps} \Pr[Y' \in S] + \delta,\\
    \mbox{and}\,\,\,\,     \Pr[Y' \in S] &\leq e^{\eps} \Pr[Y \in S] + \delta.
  \end{align*}
  When two r.v.'s are $(\eps,0)$-indistinguishabile, we will often call them $\eps$-indistinguishable.
\end{definition}

\begin{definition}[Differential Privacy]
  A mechanism $M: D^n \rightarrow R$ is $(\eps,\delta)$-differentially private if for any pair of neighboring datasets $\mathbf{x},\mathbf{x'}$, the random variables $M(\mathbf{x})$ and $M(\mathbf{x'})$ are $(\eps, \delta)$-indistinguishable.
\end{definition}
Here neighboring datasets typically means datasets that differ in the input of one user.

We consider a data stream $\x = x_1,\ldots, x_T$, where each $x_i \in D$. In this work we will be concerned with the case that $D = \{0,1\}$.  A streaming algorithm is defined by a set of functions. The function $\initialize$ (usually left implicit) sets up the state on device, including the state $s_0$, and potentially some state on the server to allow coordination between the two. We have a sequence of functions, where $\textbf{State}_t : D\times \internalstatespace \to \internalstatespace$ is a (possibly randomized) function that takes the input at time $t$ and the current state $s_{t-1}$, and maps to the new state. Given an input stream $x_1,\ldots,x_T\in D$ and an initial state $s_0$ (we will sometimes omit $s_0$ as an argument for brevity), let $s_t=\textbf{State}_t(x_t;s_{t-1})$ and $\textbf{State}(x_1,\ldots,x_T)=(s_0,s_1,\ldots,s_T)$.

An estimation algorithm is defined by additional functions $\Out$ and $\agg$, where $\Out$ maps $s_T$ to an output space $\OS$, and $\estimate$ takes a vector of $n$ elements from $\OS$ and computes an estimate of the desired statistic.

\begin{definition} [Local pan-privacy]
  We say a streaming algorithm defined by a set of functions $\textbf{State}_t$ is $\eps$-locally pan-private if for any pair of streams $\x = x_1,\ldots,x_T$ and $\x'=x'_1,\ldots,x'_T$, the state vectors $\textbf{State}(\x)$ and $\textbf{State}(\x')$ are $\eps$-indistinguishable.
\end{definition}

\begin{remark}
    \label{remark:definition}
    Several remarks are in order.  (a) The definition considers neighboring datasets to differ in the full stream at one device. Thus it provides user-level privacy. (b) We require that privacy holds against an adversary that can monitor the internal state on the device after each event, and all communication out of the device. In other words, we guard against multiple intrusions, or continuous intrusion. The problem becomes easier if we only had to guard against a single intrusion. However, for the examples that motivate this work (shared device such a public computer), restricting the adversary to a single intrusion is unnatural. (c)  A natural generalization of our model can allow the device to send multiple messages, while making sure to account for the potential information leakage from the event of sending or not-sending a message (that may be observable by an adversary). This generalization does not make the model stronger, as an algorithm that is locally pan-private in this model can store the messages it would have sent, and send them at step $T$.\footnote{If we were concerned about the size of $\internalstatespace$, the many-messages version of the model may be more powerful.} (d) In some settings, one may only be interested in a smaller set of {\em admissible} streams $\x$, and one can naturally adapt the definition above by requiring the streams $\mathbf{x}$ and $\mathbf{x'}$ to be admissible. We have aimed to keep the definition above simple, at the expense of generality. One can replace the quantifier over $\x, \x'$ above to require them to lie in some set $\mathcal{X}$ of admissible streams.
\end{remark}
The following definitions capture the standard notions of privacy for such algorithms.
\begin{definition}
  We say an estimation algorithm is $(\eps,\delta)$-locally differentially private if for any pair of streams $\x=x_1,\ldots,x_T$ and $\x'=x'_1,\ldots,x'_T$, the distributions $\Out(\textbf{State}(\x))$ and $\Out(\textbf{State}(\x'))$ are $(\eps,\delta)$-indistinguishable.
\end{definition}
\begin{definition}
  We say an estimation algorithm is $(\eps,\delta)$-aggregator differentially private if for any pair of streams $\x$ and $\x'$, and any set of $(n-1)$ streams $\{\y^{(j)}\}_{j=1}^{n-1}$ the distributions defined by $\sum_{i=1}^{n-1} \Out(\textbf{State}(\y^{(i)})) + \Out(\textbf{State}(\x))$ and $\sum_{i=1}^{n-1} \Out(\textbf{State}(\y^{(i)})) + \Out(\textbf{State}(\x'))$
 are $(\eps,\delta)$-indistinguishable.
\end{definition}


We next define computational indistinguishability. A “security” parameter $\lambda$ will control various quantities in these definitions. The adversary will be computationally bounded to be polynomial in $\lambda$, and we say a function in $\lambda$ is {\em negligible} if it approaches zero faster than the inverse of any polynomial in $\lambda$.
\begin{definition}
Let $\{D_{\lambda}\}_{\lambda}$ and $\{D'_{\lambda}\}_{\lambda}$ be two ensembles of distribution. We say that they are $f(\lambda)$-computationally indistinguishable if for any non-uniform probabilistic polynomial time algorithm $A$, 
\begin{align*}
    |\Pr_{z \sim D_\lambda}[A(z) = 1] - \Pr_{z \sim D'_\lambda}[A(z) = 1]| \leq f(\lambda).
    \end{align*}
When two ensembles are $f(\lambda)$-computationally indistinguishable for negligible function $f$, we say that they are computationally indistinguishable.
\end{definition}

We use a public-key encryption scheme. 

\begin{definition}
\label{def:pke}
  A public key encryption scheme is defined by the following set of probabilistic polynomial time (p.p.t.) algorithms.
  \begin{description}
  \item[Key Generation] $\KeyGen(\cdot)$ takes a security parameter in unary $1^\lambda$ and outputs a pair of keys $(k_{priv}, k_{pub})$, each in $\{0,1\}^*$.
  \item[Encryption] $\enc(m, k_{pub})$ takes a message $m$ and a public key $k_{pub}$ and outputs an encryption $c \in \{0,1\}^*$.
    \item[Decryption] $\dec(c, k_{priv})$ takes a ciphertext $c$ and a private key $k_{priv}$ and outputs a message $m$.
    \end{description}

 These functions have the following properties:
 \begin{description}
 \item[Correctness] Suppose that $(k_{priv}, k_{pub}) \leftarrow \KeyGen(1^\lambda)$ for some $\lambda$. Then for any $m \in \{0,1\}^k$, it holds that $\dec(\enc(m, k_{pub}), k_{priv})$ equals $m$.
   \item[Semantic Security] For any $m, m'$, the encryptions $\enc(m, k_{pub})$ and $\enc(m', k_{pub})$ are computationally indistinguishable.
 \end{description}
\end{definition}

We rely on encryption schemes that allow re-encryption of an encrypted message.\footnote{This is slightly weaker than the usual notion of rerandomizable encryption, as we allow $c$ and $\rerandomize(c, k_{pub})$ to be distinguishable. We show that this notion is necessary and sufficient for our purposes.} Informally, a rerandomizable encryption allows for an encrypted message to be ``rerandomized'', i.e. to be replaced by a new encryption of the same message, that is indistinguishable from a new encryption. Here and in the rest of the paper $\rerandomize^t(\cdot)$ denotes the $t$-fold composition $\rerandomize(\rerandomize(\ldots(\rerandomize(\cdot))))$.
\begin{definition}\label{def:rerandomizable}
    A public key encryption scheme is {\em rerandomizable} if there is a function $\rerandomize(c, k_{pub})$ with the following properties. For any $m,m'$, and any $t\geq 0$, let $c=\rerandomize^{t}(\enc(m, k_{pub}))$ and $c'=\rerandomize^{t}(\enc(m', k_{pub}))$, and let $\bar{c} =\rerandomize(c, k_{pub})$ and $\bar{c'} = \rerandomize(c', k_{pub})$. Then $\dec(c, k_{priv}) = \dec(\bar{c}, k_{priv})$. Furthermore, the tuple $(c, \bar{c})$ is computationally indistinguishable from $(c, \bar{c'})$. 
\end{definition}

Note that this definition implies that for any $t$ and any $m,m'$, the distributions $\rerandomize^t(\enc(m))$ and $\rerandomize^{t}(\enc(m'))$ are computationally indistinguishable. We remark that we do not require a rerandomized ciphertext to be indistinguishable from a freshly generated one: it need only be indistinguishable from a rerandomized encryption of any different plaintext (for the same number of rerandomizations $t$).

\citet{MironovPRV09} defined and related different notions of computational differential privacy, and those notions can be extended to local pan-privacy. We state a version of this definition next.
\begin{definition} [Computational $(\eps,\delta)$-indistinguishability]
\label{def:comp_dp}
Let $\{D_{\lambda}\}_{\lambda}$ and $\{D'_{\lambda}\}_{\lambda}$ be two ensembles of distribution. We say that they are computationally $(\eps,\delta)$-indistinguishable if for any non-uniform probabilistic polynomial time algorithm $A$, 
\begin{align*}
    \Pr_{z \sim D_\lambda}[A(z) = 1] &\leq e^{\eps}\cdot\Pr_{z \sim D'_\lambda}[A(z) = 1]| + \delta + \negl(\lambda).
    \end{align*} 
  When two ensembles of r.v.'s are computationally $(\eps,0)$-indistinguishabile, we will often call them computationally $\eps$-indistinguishable.
    
\end{definition}
\begin{definition}[Computational Local Pan-Privacy,{\sc Ind-Cdp} version]
  We say a streaming algorithm defined by a set of functions $\textbf{State}_t$ is computationally $\eps$-locally pan-private if for any pair of streams $\x = x_1,\ldots,x_T$ and $\x'=x'_1,\ldots,x'_T$, the state vectors $\textbf{State}(\x)$ and $\textbf{State}(\x')$ are computationally $\eps$-indistinguishable.

\end{definition}

We recall some basic mechanisms that will be useful. We refer the reader to~\citet{DworkR14} for their privacy proofs.
\begin{definition}[Randomized Response]
Let $\eps > 0$. The randomized response mechanism $\RR{\eps} : \{0,1\} \rightarrow \{0,1\}$ is an $\eps$-DP mechanism defined as:
\begin{equation*}
    \RR{\eps}(b) = \left\{ 
    \begin{array}{ll}
    b & \mbox{with probabilty } \tfrac{e^{\eps}}{1+e^{\eps}}\\
    1-b & \mbox{with probabilty } \tfrac{1}{1+e^{\eps}}\\
    \end{array}
    \right.
\end{equation*}

\end{definition}
The following observation about implementing randomized response will be useful.

\begin{observation}
    \label{obs:rr_implementation}  Let $b \in \{0,1\}$ and $\eps > 0$. Then $\RR{\eps}(b)$ can be implemented by outputting $b$ with probability $\frac{e^{\eps}-1}{e^{\eps}+1}$, and outputting a random bit otherwise.
\end{observation}

We will use the following lemma that says that any $2$-input local randomizer is a post-processing of a randomized response.
\begin{lemma}\cite{Kairouz:2015}
\label{lem:its_2RR}
    Let $\lr:\{0,1\}\to\internalstatespace$ be an $\epsilon$-DP local randomizer. Then there exists a post-processing function $h:\{0,1\}\to\internalstatespace$ such that $\lr(0)=h(\RR{\epsilon}(0))$ and $\lr(1)=h(\RR{\epsilon}(1))$.
\end{lemma}

The following utility bounds for randomized response are standard~\citep{FeldmanMT:2020}.
\begin{theorem}
\label{thm:rr_local}
    Let $\eps_0>0$. Then for any $b_1,\ldots,b_n \in \{0,1\}$, one can post-process their randomized responses $\{y_i = \RR{\eps_0}(b_i)\}$ to derive an estimate $\hat{S}$ of $S=\sum_i b_i$ such that
    \begin{align*}
        \E[\hat{S}]=S; &\;\;\;\;\; \E[|\hat{S}-S|] \leq O\left(\sqrt{n\left(1+\tfrac{e^{\eps_0}}{(1+e^{\eps_0})^2}\right)}\right).
    \end{align*}
\end{theorem}
\begin{theorem}
\label{thm:rr_agg}
    Let $(\eps,\delta)\in (0,1)$. Then there is an $\eps_0$ such that for any $b_1,\ldots,b_n \in \{0,1\}$, the randomized responses $\{y_i = \RR{\eps_0}(b_i)\}$ are $\eps$-DP in the aggregator model. Further their sum can be post-processed to derive an estimate $\hat{S}$ of $S=\sum_i b_i$ such that
    \begin{align*}
        \E[\hat{S}]=S; &\;\;\;\;\; \E[|\hat{S}-S|] \leq O(\sqrt{\log\tfrac 1 \delta}/\eps).
    \end{align*}
\end{theorem}

Let $\x = x_1,\ldots,x_T$ be an input stream with each $x_i \in \{0,1\}$. We denote by $\one(\x)$ the predicate $(\max_{i} x_i = 1)$. For a set of streams $\{\x^{(i)}\}_{i=1}^n$, the \countnonzero\ value is defined as $\sum_i \one(\x^{(i)})$.
\section{Lower Bound}
\label{sec:lb}

In this section, we prove a lower bound showing that information-theoretic local pan-privacy incurs a non-trivial cost on the achievable accuracy for the basic \countnonzero\ task. Without local pan-privacy, this task can be solved easily. Each device maintains $\one(\x_{1:t})$, which can be done with a single bit of state. Randomized response on this state after $T$ steps allows us to estimate \countnonzero\ with additive error $\tilde{O}(1/\eps)$ in the aggregator DP setting (\cref{thm:rr_agg}), and $\tilde{O}(\sqrt{n(1+ \frac{e^{\eps}}{(e^{\eps}-1)^2})})$ in the local DP setting (\cref{thm:rr_local}). In contrast, we show that information-theoretic local pan-privacy entails a polynomial dependence on $T$. This lower bound holds for the case when the input streams are restricted to have at most a single `1'.

Each client receives a stream of inputs $x_1, \ldots, x_T \in\{0,1\}$, can communicate once to the server, and the goal of the server is to compute the number of clients such that $\one(\x)=1$. We will prove the following theorem. Note that this means that the correlation between the message a device sends, and $\one(\x)$ is at most $O(1/\sqrt{T})$

\begin{theorem}
\label{thm:lower_bound}
    Let $\mathcal{A}$ be an $\eps$-locally pan-private client-side algorithm for the \countnonzero\ task comprising of the pair $\textbf{State}$ and $\textbf{Out}$. Then there exists inputs $\x$ and $\x'$ such that $\one(\x)=0$, $\one(\x') = 1$, and 
    \begin{align*}
        TV(\textbf{Out}(\textbf{State}(\x)), \textbf{Out}(\textbf{State}(\x'))) \leq O((e^{\eps}-1)/\sqrt{{Te^\eps}}).
    \end{align*}
\end{theorem}
This bound, coupled with standard techniques (e.g. ~\citep{DuchiJW13}), yields the following lower bound. This shows that for large $T$, local pan-privacy must come at a significant cost.
\begin{theorem}
\label{thm:main_lb}
    Let $\mathcal{A}$ be an $\eps$-locally pan-private algorithm run on $n$ clients, and let $\hat{S}$ be any estimate of \countnonzero$(\x^{(1)},\ldots,\x^{(n)})$ derived from the the outputs of $\mathcal{A}$. For $T \geq 16e^\eps$, $\hat{S}$ has expected error $\Omega(\sqrt{nTe^{\eps}/(e^{\eps}-1)^2})$.
\end{theorem}

\begin{figure}[ht]
\begin{algorithm}[H]
\caption{$\distinguisher$}\label{distinguisher}
    \begin{algorithmic}[1]
    \State \textbf{Require:}  $\Statefn$, $\out$, initial state $s_0$
    \State \textbf{Input:} $b$
    \If{$b=0$}
        \State \textbf{Return} $\out(\Statefn((0,\ldots,0);s_0))$
    \Else
        \State Select $\tilde{t}\sim \texttt{Uniform}([T])$
        \State Set $x_{\tilde{t}}=1$ and $x_t=0$ for $t\in[T]\backslash\{\tilde{t}\}$
        \State \textbf{Return} $\out(\Statefn((x_1,\ldots,x_T);s_0))$
    \EndIf
    \end{algorithmic}
\end{algorithm}
\end{figure}
 We will now prove \cref{thm:lower_bound}. We start by using the locally pan-private algorithm $\mathcal{A}$ to construct a randomized function \distinguisher\ from $\{0,1\}$ to the output space of the algorithm. On input `0', this algorithm runs $\mathcal{A}$ on the all $0$'s stream. On input `1', it will put a $1$ at a uniformly random place in the stream (with other $x_t$'s being $0$) and run $\mathcal{A}$ on the resulting stream.  Since $\one(\x)$ is different on these two streams, we expect this \distinguisher\ to behave differently enough on $0$ and $1$. We will argue that the local pan-privacy constraint prevents these distributions from being too far.

Towards this goal, we first argue that the output of a locally pan-private algorithm can be obtained as a post-processing of randomized responses of individual $x_t$'s.
\begin{lemma}\label{postprocessing}
    Given an initial state $s_0$, there exists a post-processing function $g$ such that 
    \[\Statefn((x_1,\ldots,x_T);s_0)=g(\RR{\epsilon}(x_1), \ldots, \RR{\epsilon}(x_T)).\]
\end{lemma}

\begin{proof} 
We will first show that if $\Statefn$ is $\epsilon$-DP then for all $t\in[T]$ and any state $s_{t-1}$, $\Statefn_t(\cdot;s_{t-1})$ is $\epsilon$-DP. Let $t\in[T]$, $s_{t-1}\in\internalstatespace$, and $x_t,x_t'\in\{0,1\}$. Further, let $\{x_{t'}\}_{t'\in[T]\backslash\{t\}}\in\{0,1\}^{T-1}$ and $\{s_{t'}\}_{t'\in[T]\backslash\{t\}}\in\internalstatespace^{T-1}$, then
\begin{align*}
    \frac{\Pr(\Statefn_t(x_t;s_{t-1})=s_t)}{\Pr(\Statefn_t(x_t';s_{t-1})=s_t)} &= \frac{\Pr(\Statefn_t(x_t;s_{t-1})=s_t)}{\Pr(\Statefn_t(x_t';s_{t-1})=s_t)}\prod_{t'\in[T]\backslash \{t\}}\frac{\Pr(\Statefn_{t'}(x_{t'};s_{t'-1})=s_{t'})}{\Pr(\Statefn_{t'}(x_{t'};s_{t'-1})=s_{t'})}\\
    &= \frac{\Pr(\Statefn((x_1,\ldots,x_t,\ldots,x_T);s_0)=(s_1,\ldots,s_T))}{\Pr(\Statefn((x_1,\ldots,x_t',\ldots,x_T);s_0)=(s_1,\ldots,s_T))}\\
    &\le e^{\epsilon}.
\end{align*}
Thus the map $\Statefn_t(\cdot;s_{t-1})$ can be viewed as a local DP mechanism, and thus \cref{lem:its_2RR} implies that it can be written as a postprocessing of a randomized response.
Given $t\in[T]$ and $s_{t-1}\in\internalstatespace$, let $h_{t,s_{t-1}}:\{0,1\}\to\internalstatespace$ be such that for $b\in\{0,1\}$, $\Statefn_t(b;s_{t-1})=h_{t,s_{t-1}}(\RR{\epsilon}(b))$. Then we can define $g(b_1, \ldots, b_T)=(s_1,\ldots,s_T)$ where $s_t=h_{t,s_{t-1}}(b_t)$.
\end{proof}

We next argue that this, and the fact that our distributions on $\x$ are exchangeable implies that the output really is a post-processing of the {\em count} of the number of $1$s in the randomized responses.
\begin{lemma}
\label{lem:tv_small}
    Given any initial state $s_0$, there exists a post-processing function $h$ such that $\distinguisher(0)=h(\Bin(T,\frac{1}{e^{\epsilon}+1}))$ and $\distinguisher(1)=h(\Bin(T-1,\frac{1}{e^{\epsilon}+1})+\Ber(\frac{e^{\epsilon}}{e^{\epsilon}+1}))$. Further, 
    \begin{align}
    &\TV(\distinguisher(0), \distinguisher(1))\\&\;\;\;\;\le \TV\left(\small{\Bin(T,\tfrac{1}{e^{\epsilon}+1}),\Bin(T-1,\tfrac{1}{e^{\epsilon}+1})+\Ber(\tfrac{e^{\epsilon}}{e^{\epsilon}+1})}\right) \label{eqn:tv_f_1}\\
    &\;\;\;\;=O\left((e^{\eps}-1)\sqrt{\frac 1 {Te^{\eps}}}\right). \label{eqn:tv_f_2}
    \end{align}
\end{lemma}

\begin{proof}
By Lemma~\ref{postprocessing}, there exists $g:\{0,1\}^T\to\internalstatespace^T$ such that $\Statefn((x_1,\ldots,x_T);s_0)$ can be written as $g(\RR{\epsilon}(x_1),\ldots,\RR{\epsilon}(x_T))$. Under our distribution on $\x$ conditioned on $b$, the random variables $x_t$'s are exchangeable, and thus conditioned on the count of $1$s in $\RR{\eps}(x_i)$'s, all permutations are equally likely. The function $h$ is then described in Algorithm~\ref{postprocessdistinguisher}. This implies the first part of the claim. The inequality \cref{eqn:tv_f_1} is a consequence of the data processing inequality.  
\begin{figure}[ht]
\begin{algorithm}[H]
\caption{$h$}\label{postprocessdistinguisher}
    \begin{algorithmic}[1]
    \State \textbf{Require:} $g$, $\out$
    \State \textbf{Input:} $m\in[T]$
    \State Sample $\mathbf{b}$ uniformly at random from $\{\mathbf{b}\in\{0,1\}^T\;|\; |\mathbf{b}|_1=m\}$.
    \State \textbf{Return} $\out(g(\mathbf{b})).$
    \end{algorithmic}
\end{algorithm}
\end{figure}
The second inequality is standard. For completeness, we give a proof in \cref{app:binomials_tv}.

\end{proof}

\cref{lem:tv_small} implies \cref{thm:lower_bound} and completes the proof of our lower bound.

\section{The Single-Server Model}
\label{sec:ub-single}
In this section, we will discuss the single-server model. We will require a rerandomizable public-key encryption scheme. Our computational local pan-privacy will be achieved by the device storing encryptions of any sensitive state.  As the device does not hold the private key, the on-device state is computationally indistinguishable from a sequence of encryptions of 0.

\subsection{Counting Devices that have at least one occurrence}

Let us first consider the problem of differentially privately computing the number of clients for which the sensitive event occurs at least once. In Section~\ref{sec:histograms} we will discuss how to extend our approach to building a histogram of the number of occurrences of the sensitive event. 
We start by describing how the count will be stored and updated on device. At the beginning of the collection period, the client initializes the state to be an encryption of 0. At each time step, the device updates the state. If $x_t$ is 1, i.e. if the sensitive event has occurred in this time step, the device replaces the state with a fresh encryption of 1, with an appropriate number of rerandomizations applied. If $x_t$ is 0, the device rerandomizes the current state. This continues until the end of the time horizon. Pseudo-code is given in Algorithm~\ref{localpanprivcounting}. This algorithm does not reveal anything about whether or when updates have occurred to an intruder, even if the intruder can view the internal state of the device after each time step. The intrusion resistance comes from the fact that an intruder can not distinguish between a fresh (rerandomized) encryption of 1, and a rerandomization a current encrypted value. Formally, the state after $t$ steps is computationally indistinguishable from $\rerandomize^{t}(\enc(0))$.  
Finally, the device uses randomized response to send their encrypted bit to the server. This can be done without needing to decrypt the state using~\cref{obs:rr_implementation}. The sent message is computationally indistinguishable from $\rerandomize^{T+1}(\enc(0))$. We remark that for most practical cryptosytems, $\rerandomize^t(\enc(b))$ has the same distribution as $\enc(b)$ so that we can  optimize the algorithm by replacing the $\rerandomize^t(\enc(b))$ steps by $\enc(b)$.
\begin{figure}[ht]
\begin{algorithm}[H]
\caption{\countnonzero, Client Algorithm}\label{localpanprivcounting}
    \begin{algorithmic}[1]
    \Require $T, \epsilon_0, \enc, \rerandomize, \mathbf{x} = x_1,\ldots,x_T$
    \State \underline{\textbf{Initialization}}
    \State $c=\enc(0)$
    \vspace{0.1in}
    \State \underline{\textbf{State Update}}
    \For{$t=1:T$}
        \If{$x_t = 1$}
            \State $c=\rerandomize^t(\enc(1))$
        \Else
            \State $c=\rerandomize(c)$
        \EndIf
    \EndFor
    \vspace{0.1in}
    \State \underline{\textbf{Send to Server}}
    \State $r={\rm Ber}(\frac{e^{\epsilon_0}-1}{e^{\epsilon_0}+1})$
    \If{$r=0$}
        \State $r'={\rm Ber}(1/2)$
        \If{$r'=1$}
            \State $c=\rerandomize^{T+1}(\enc(0))$
        \Else
            \State $c=\rerandomize^{T+1}(\enc(1))$
        \EndIf
    \Else
        \State $c = \rerandomize(c)$
    \EndIf
    \State \textbf{Return} $c$
    \end{algorithmic}
\end{algorithm}
\end{figure}

At the end of the collection period, the goal of the server is to compute a differentially private estimate of the number of devices for which the sensitive event occurred at least once. Since they have the private key, they can decrypt the local reports from each client, then aggregate and de-bias the resulting sum in the same way they would for randomized response. The following result follows:

\begin{theorem}\label{thm:ub_countnonzero}
    Let $\eps_0>0$. Suppose that each client $i$ uses \cref{localpanprivcounting} on its input $\x^{(i)}$. Then the server can estimate \countnonzero\ on inputs $\{\x^{(i)}\}_{i=1}^n$ with expected error $O\left(\sqrt{n\left(1+\tfrac{e^{\eps_0}}{(1+e^{\eps_0})^2}\right)}\right)$ and $\eps_0$-local differential privacy. For any $(\eps,\delta) \in (0,1)$, there is an $\eps_0$ such that the mechanism is $(\eps,\delta)$-aggregator DP, and has expected error $O(\sqrt{\log\tfrac 1 \delta}/\eps)$. The client algorithm satisfies computational $0$-local pan-privacy.
\end{theorem}

\subsection{Computing Histograms of the Number of Occurrences}
\label{sec:histograms}
In this section we will consider how to compute a histogram of the number of clients that have a particular number of occurrences. We will solve a slightly more general problem, where we are given a $k$ and the goal is to estimate for each $i \in \{0,1,\ldots, k-1\}$ the number of devices with count $|\mathbf{x}|_1$ equal to $i$, as well as the number of devices with count at least $k$. While one could always set $k=T$, this generalization can allow more efficient solutions in the regime where $T$ is large and the number of occurrences per device is much smaller than $T$. To handle this last bucket of count at least $k$, we will maintain an (encrypted) indicator $d_i$ for each $i$, of the event $|\mathbf{x}|_1 \geq i$.

At the beginning of the collection period, the device initializes $k+1$ counters to be an encryption of the initial histogram, i.e. $c_0=\enc(1), c_1=\enc(0), \ldots, c_k=\enc(0)$. Additionally, it initializes $d_0 = \enc(1)$ and $d_i = \enc(0)$ for each $i=1,\ldots,k$. At each time step, if an event has not occurred, all the counters are rerandomized. If an event has occurred then all the counters are rerandomized \emph{and} shifted by 1. The counter $c_0$ is set to be an $\enc(0)$, and $d_0$ is set $\enc(1)$. At the end of $T$ steps, the device has the desired values $(v_0,v_1,\ldots,v_k)$ in $(c_0,c_1,\ldots,c_{k-1},d_k)$, since $c_i$ encrypts $\one(|\x|_1=i)$ and $d_k$ encrypts $\one(|\x|_1 \geq k)$. It uses randomized response on each of these to send the result to the server.
Pseudo-code is given in Algorithm~\ref{alg:histograms}.

\begin{figure}[!h]
\begin{algorithm}[H]
\caption{Counter, Client Algorithm}\label{alg:histograms}
    \begin{algorithmic}[1]
    \Require $k, T, \enc, \rerandomize, \x = x_1,x_2,\ldots,x_T$.
    \State \underline{\textbf{Initialization}}
    \State $c_0=\enc(1)$; $d_0=\enc(1)$
    \For{$i=1:k$}
        \State $c_i=\enc(0)$; $d_i = \enc(0)$
    \EndFor
    \vspace{0.1in}

    \State \underline{\textbf{State Update}}
    \For{$t=1:T$}
       \If{$x_t = 0$}
        \For{$i=0:k$}
            \State $c_i=\rerandomize(c_i)$, $d_i=\rerandomize(d_i)$.
        \EndFor
       \Else
            \State $c_0=\rerandomize^{t}(\enc(0))$; $d_0=\rerandomize^t(\enc(1))$
            \For{$i=1:k$}
                \State $c_i=\rerandomize(c_{i-1})$; $d_i=\rerandomize(d_{i-1})$
            \EndFor
        \EndIf
    \EndFor
    \vspace{0.1in}

    \State \underline{\textbf{Send to Server}}
    \For{$i=0:k-1$}
        \State $v_i=\rerandomize(c_i)$
    \EndFor
    \State $v_k = \rerandomize(d_k)$
    \For{$i=0:k$}
        \State $r={\rm Ber}(\frac{e^{\epsilon_0}-1}{e^{\epsilon_0}+1})$
        \If{$r=0$}
            \State $r'={\rm Ber}(1/2)$
            \If{$r'=1$}
                \State $v_i=\rerandomize^{T+1}(\enc(0))$
            \Else
                \State $v_i=\rerandomize^{T+1}(\enc(1))$
            \EndIf
        \EndIf
    \EndFor

    \State \textbf{Return} $(v_0, \ldots, v_k)$
    \end{algorithmic}
\end{algorithm}
\end{figure}

Upon receiving the reports from each client, the server can decrypt the randomized responses. We sum these reports and use the standard de-biasing for randomized response to obtain an unbiased estimate of each histogram bucket count. Note that the sensitivity of the vector $(v_1,\ldots, v_k)$ is $1$ as at most one of the values can be $1$. Thus even though $k$ randomized responses are sent, the noise added does not grow with $k$. Since we run randomized response on each coordinate of the histogram, \cref{thm:rr_local} and \cref{thm:rr_agg} implies that
\begin{theorem}\label{thm:ub_histogram}
    \cref{alg:histograms} is computationally $0$-locally pan-private, and $\eps_0$-local DP with respect to the server. It estimates the histogram with expected error $O\left(\sqrt{n\left(1+\tfrac{e^{\eps_0}}{(1+e^{\eps_0})^2}\right)}\right)$ for each bucket count. Moreover, for any $(\eps,\delta) \in (0,1)$, there is an $\eps_0$ such that the mechanism is $(\eps,\delta)$-aggregator DP, and has expected error $O(\sqrt{\log\tfrac 1 \delta}/\eps)$ for each bucket count.
\end{theorem}

\subsection{Computing the Average Number of Occurrences per Device}

Computing the average number of occurrences per-device can be done using histograms, under the assumption that no count is larger than $k$. Alternatively, the resulting average can be viewed as an average of a truncation (at $k$) of the original counts. However using the estimate
   $ \sum_{j} c_j = \sum_{i=1}^k i \cdot |\{j: c_i = j\}|$
will result in error that scales as $O(k^{\frac 3 2})$. 
We next show that we can reduce the error to $O(k)$, matching the bound one would get in the central model.
\input{Sections/averages}

\section{The Two Server Model}
\label{sec:ub-two-server}
In this section we will address the two-server model of~\cite{Corrigan-GibbsB17}, where the client sends secret-shares of their data to two servers. Any sensitive information stored locally will be secret-shared with the shares encrypted by the public keys of two separate servers at all times. This scheme is protected against continuous intrusion by an adversary that does not collude with either of the servers.

In the two-server model, we often want to additionally protect the server against malicious clients who seek to poison the result by sending secret-shares of invalid reports. For example, instead of sending a secret-sharing of 0 or 1, a malicious device may send a secret-sharing of a million to skew the result. We will build on zero-knowledge proofs of validity for the predicates of interest. Our construction will maintain encrypted secret-shares of the state. The additional complexity arising from these proofs is that to construct the proof of validity, one needs to know the secret-shared data. To address this, we make use of an important fact. The proofs in Prio depend on the input being secret-shared (as they must), but conditioned on the input, do not depend on the encryptions of the secret-shares. This allows the proofs to remain valid when we rerandomize the encryptions. Thus our algorithms will generate proofs of validity when the relevant secret-shares are created. One can rerandomize the encryptions of the secret-shares and the proofs, without impacting the validity of the proofs.

\subsection{Counting Devices that have at least one occurrence}
We first consider the problem of counting the number of devices that have at least one occurrence in the two-server model. The on-device and server-side algorithms will both be very similar to the single-server setting but with the addition of secret-sharing of any sensitive information on device, and creating validity proofs. We assume access to function \secretshare\ that generates a random secret-sharing, and a function \validityproof\ that creates a non-interactive zero-knowledge proof that the secret-shares represent a value in $\{0,1\}$. The pseudocode for the algorithm is in \cref{localpanprivcounting2s}.

We remark that this algorithm can resist an intrusion by an adversary that does not collude with either of the two servers. This can be relaxed to allow the adversary to collude with one of the two servers. The place where the collusion with a server can create a challenge is that a colluding server can help the adversary detect whether its share is changed in step \ref{step:newshare} or if only the encryption is rerandomized and the share itself is unchanged (step~\ref{step:rerandomize}). If we use an additively homomorphic encryption scheme (see \cref{def:phe}), we can create new secret-shares under the encryption instead, so that each $c^{(i)}$ is an encryption of a fresh random field element after each update. This allows for non-interactive proofs, including the SNIPs \citep{Corrigan-GibbsB17} that use Beaver triples.
 \algrenewcommand\algorithmicindent{1.0em}%
 \begin{figure}[!h]
\begin{algorithm}[H]
\caption{\countnonzero, Client Algorithm, Two-server model}\label{localpanprivcounting2s}
    \begin{algorithmic}[1]
    \Require $T, \enci{1}, \rerandomizei{1}, \enci{2}, \rerandomizei{2}, \x=x_1,\ldots,x_T$
    \State \underline{\textbf{Initialization}}
    \State $(\si{1},\si{2})=\secretshare(0)$
    \State $(\pi^{(1)}, \pi^{(2)}) = \validityproof(\si{1},\si{2})$
    \State $(\ci{1}, \ci{2}) = (\enci{1}(\si{1}), \enci{2}(\si{2}))$
    \State $(\pri{1},\pri{2}) = (\enci{1}(\pi^{(1)}),  \enci{2}(\pi^{(2)}))$
    
    \vspace{0.1in}

    \State \underline{\textbf{State Update}}
    \For{$t=1:T$}
        \If{$x_t = 1$}
            \State $(\si{1},\si{2})=\secretshare(1)$ \label{step:newshare}
  \State $(\pi^{(1)}, \pi^{(2)}) = \validityproof(\si{1},\si{2})$
    \State $(\ci{1}, \ci{2}) = (\rerandomizei{1}^t(\enci{1}(\si{1})), \rerandomizei{2}^t(\enci{2}(\si{2})))$
    \State $(\pri{1},\pri{2}) = (\rerandomizei{1}^t(\enci{1}(\pi^{(1)})),  \rerandomizei{2}^t(\enci{2}(\pi^{(2)})))$    
        \Else
    \State $(\ci{1}, \ci{2}) = (\rerandomizei{1}(\ci{1}), \rerandomizei{2}(\ci{2}))$   \label{step:rerandomize} 
    \State $(\pri{1}, \pri{2}) = (\rerandomizei{1}(\pri{1}), \rerandomizei{2}(\pri{2}))$
    \EndIf
    \EndFor
    \vspace{0.1in}

    \State \underline{\textbf{Send to Server}}
    \State $r={\rm Ber}(\frac{e^{\epsilon_0}-1}{e^{\epsilon_0}+1})$
    \If{$r=0$}
        \State $r'={\rm Ber}(1/2)$
        \If{$r'=1$}
    \State $(\si{1},\si{2})=\secretshare(0)$
  \State $(\pi^{(1)}, \pi^{(2)}) = \validityproof(\si{1},\si{2})$
    \State $(\ci{1}, \ci{2}) = (\rerandomizei{1}^t(\enci{1}(\si{1})), \rerandomizei{2}^t(\enci{2}(\si{2})))$
    \State $(\pri{1},\pri{2}) = (\rerandomizei{1}^t(\enci{1}(\pi^{(1)})),  \rerandomizei{2}^t(\enci{2}(\pi^{(2)})))$    
        \Else
    \State $(\si{1},\si{2})=\secretshare(1)$
      \State $(\pi^{(1)}, \pi^{(2)}) = \validityproof(\si{1},\si{2})$
    \State $(\ci{1}, \ci{2}) = (\rerandomizei{1}^{T}(\enci{1}(\si{1})), \rerandomizei{2}^T(\enci{2}(\si{2})))$
    \State $(\pri{1},\pri{2}) = (\rerandomizei{1}^{T}(\enci{1}(\pi^{(1)})),  \rerandomizei{2}^T(\enci{2}(\pi^{(2)})))$    
        \EndIf
      \EndIf
    \State $(\ci{1}, \ci{2}) = (\rerandomizei{1}(\ci{1}), \rerandomizei{2}(\ci{2}))$
    \State $(\pri{1}, \pri{2}) = (\rerandomizei{1}(\pri{1}), \rerandomizei{2}(\pri{2}))$
  
    \State \textbf{Return} $(\ci{1}, \pri{1}, \ci{2}, \pri{2})$
    \end{algorithmic}
\end{algorithm}
\end{figure}
\algrenewcommand\algorithmicindent{2.0em}%

Each client sends the secret-shares of their contribution to the two servers. Each server can then decrypt the shares with their private key and aggregate the result. The sum of the shares at the two servers can then be de-biased to give the final result. The servers can also validate the proofs to ensure that the secret-shared values are all in $\{0,1\}$. 

\subsection{Computing Histograms of the Number of Occurrences}

Using the same approach as we did for \countnonzero, we can extend our algorithm for histograms in the single-server setting to work in the two-server model as well. As with the \countnonzero\ algorithm, the validity property that is being validated is that each of the $k+1$ reported values are in $\{0,1\}$. Thus the same approach to constructing validity proofs suffices. For brevity, we omit a detailed description of the algorithm.

\section{On the need for rerandomizable Public-key cryptography}
\label{sec:pkeneeded}
We show that a rerandomizable public-key encryption scheme is necessary for computational $0$-local pan-privacy. Recall that any locally pan-private algorithm also has an $\initialize$ operation, that creates some shared state between the client and the server(s). In our implementations, this operation would be one that provides the public keys to the clients, and creates the initial state $s_0$. 
\begin{theorem}
\label{thm:pkeneeded}
    Suppose that we have a set of algorithms $\initialize, \{\Statefn_t\}_{t=1}^T, \Out$, and $\estimate$ that define a computationally $0$-locally pan-private streaming algorithm that estimates \countnonzero\ on any set of inputs with error at most $n/4$, with probability $1-\negl(n)$ for large enough $n$. Then for any security parameter $\lambda$ and given a $T=\poly(\lambda)$, we can define functions $\KeyGen, \enc, \dec$ that define a public-key encryption scheme, where each of these operations runs in $\poly(\lambda)$ time. The scheme is rerandomizable in the sense of~\cref{def:rerandomizable}, supporting up to $T$ rerandomizations.
\end{theorem}

\begin{proof}
We will build an encryption scheme that can encrypt a single bit $b$. The basic idea of the construction is to simulate running the algorithm on $n$ clients, where for each client $i$, the first input in the stream $x^{(i)}_1$ is set to $b$. The internal state of all the clients will comprise the encryption. The decryption will simulate the state transformations given $x^{(i)}_t = 0$ all the way to $T$, generating the outputs, and running $\estimate$ on the outputs. The utility guarantee implies that the decryption recovers the original intended bit. The local pan-privacy implies the security of the scheme. Finally rerandomization is achieved by simulating a single state transformation. We give details next.

For a suitably large $n=\poly(\lambda)$, $\KeyGen$ will run the $\initialize$ operation $n$ times, and return the set of $n$ client states, including the state $s_0^{(i)}$ for each client, as a public key $k_{pub}$, and the $n$ server states as private key $k_{priv}$. Additionally, the parameter $T$ will be included in both the keys. Given a bit $b$, $\enc(b, k_{pub})$ will simulate for each $i$, one step of the locally pan-private algorithm on input $b$ starting at state $s_0^{(i)}$ to derive a set of states $s^{(i)}_1$. The encryption will be defined as $(1, \{s^{(i)}_1\}_{i=1}^n)$. In general, valid ciphertexts will look like $(t,\{s^{(i)}_t\}_{i=1}^n)$, for some $t \leq T$, and some set of states  $\{s^{(i)}_t\}_{i=1}^n$. Given such an encryption, the $\dec$ algorithm will simulate, for each $i$, running the state transformations $\Statefn_{t+1},\ldots,\Statefn_{T}$ on $s^{(i)}_{t}$ with input $x^{(i)}_{t+1}=0$, followed by running $\Out$ to generate a set of $n$ messages. It then runs $\estimate$ on the collection of $n$ messages, and returns $0$ if the answer is smaller than $n/2$, and $1$ otherwise. Finally, $\rerandomize$ on a ciphertext $c=(t,\{s^{(i)}_t\}_{i=1}^n)$ (for $t<T$) will simulate, for each $i$, running the state transformations $\Statefn_{t+1}$ on  $s^{(i)}_{t}$ with input $x^{(i)}_{t+1}=0$ to get an updated state $s^{(i)}_{t+1}$. The updated ciphertext $c'$ is then set to $(t+1, \{s^{(i)}_{t+1}\}_{i=1}^n)$. Any input ciphertexts with $t\geq T$ are rejected by $\rerandomize$.

It is easy to see that $\dec(\enc(b, k_{pub}), k_{priv})$ comprises an exact simulation of running the \countnonzero\ algorithm on $n$ clients on inputs $(b,0,0,\ldots,0)$, followed by rounding the final estimate. Since the correct answer to \countnonzero\ on these inputs is $bn$ and the estimation algorithm has error $n/4$ with high probability, we conclude that the encryption scheme satisfies $\dec(\enc(b, k_{pub}), k_{priv}) = b$ except with negligible probability.

The security of the encryption  follows from the computational local pan-privacy. Indeed if we had an efficient algorithm $A$ that can distinguish $\enc(0, k_{pub})$ and $\enc(1, k_{pub})$, then by a standard hybrid argument, it can be used to distinguish $\Statefn(0, s_0)$ and $\Statefn(1,s_0)$, which violates the computational $0$-local pan-privacy.

The rerandomization is essentially simulating one step in the decoding. It is therefore immediate that for any valid ciphertext, $\dec(\rerandomize(c, k_{pub}), k_{priv})$ and $\dec(c, k_{priv})$ are running exactly the same simulation and thus are identically distributed. Finally, the security of the rerandomization follows once again from the hybrid argument: if we could distinguish iterated rerandomizations of $0$ from those of $1$, we would get a distinguisher that can learn the first input $x_1$ from the local state at some later time step $t$.
\end{proof}

We note that the proof uses instances which have either zero or one `$1$' in each input stream. Thus any accurate algorithm for histogram estimation implies one for \countnonzero\ with the same accuracy, and hence would also imply a public-key encryption scheme.

We next extend this argument to handle a large class of $\eps$-locally pan-private algorithms for $\eps>0$. This result will apply to locally pan-private algorithms that use at most logarithmic space, and which use the same state transformation function $\Statefn$ at all time steps. 

\begin{theorem}
\label{thm:pkeneededv2}
Let $\eps \in (0,1)$.
    Suppose that we have a set of algorithms $\initialize, \Statefn, \Out$, and $\estimate$ that define a computationally $\eps$-locally pan-private streaming algorithm that estimates \countnonzero\ on any set of inputs with error at most $n/4$, with probability $1-\negl(n)$ for large enough $n$, and for up to $T$ steps, using $S$ bits of space. Then for any security parameter $\lambda$, we can define functions $\KeyGen, \enc, \dec$ that define a $(\frac{\eps}{\sqrt{T}} + T\cdot \negl(\lambda))$-secure public-key encryption scheme, where each of these operations runs in $\poly(\lambda,2^S)$ time. The scheme  is rerandomizable in the sense of~\cref{def:rerandomizable}, supporting up to $T$ rerandomizations.
\end{theorem}
Note that if $T^{-1}$ is $\negl(\lambda)$ and $S$ is $O(\log \lambda)$, then the resulting public-key encryption scheme is secure. 
\begin{proof}
We will use an encryption scheme very similar to the one in the proof above. For encryption, we will set $x_t$ to one for a random $t \in \{1,2,\ldots,T/2\}$, instead of setting $x_1$ to one  in the proof of \cref{thm:pkeneeded}. The encryption algorithm simulates the state transformation up to $T/2$ steps. \cref{lem:tv_small} then allows us to argue the security of the encryption scheme. The decryption simulates the state transforms up to step $T$ and the correctness proof is as before. The additional challenge is to argue that $\enc,\dec$ run in polynomial time. This uses the fact that the result of applying the same randomized transform $\tau$ times, on a state of size $S$  can be computed in time $\exp(O(S))\cdot\log(\tau)$. We give details next.

As before, $\KeyGen$ will run the $\initialize$ operation $n$ times, and return the set of $n$ client states, including the state $s_0^{(i)}$ for each client, as a public key $k_{pub}$, and the $n$ server states as private key $k_{priv}$. Additionally, the parameter $T$ will be included in both the keys. To define $\enc$, we consider the following randomized process. Given a bit $b$, we first create $n$ independent random streams which are zero everywhere, except for a randomly picked $t_i \in \{1,\ldots,T/2\}$, where $\x^{(i)}$ is $b$. We then simulate, for each $i$, $T/2$ steps  of the locally pan-private algorithm on $\x^{(i)}$ starting at state $s_0^{(i)}$ to get $s_{T/2}^{(i)}$. $\enc(b, k_{pub})$ is then defined as $(T/2, \{s^{(i)}_{T/2}\}_{i=1}^n)$. In general, valid ciphertexts will look like $(t,\{s^{(i)}_t\}_{i=1}^n)$, for some $t \in [T/2, T]$, and some set of states  $\{s^{(i)}_t\}_{i=1}^n$. Given such an encryption, the $\dec$ algorithm will simulate, for each $i$, running the state transformations $\Statefn$ $(T-t)$ times on $s^{(i)}_{t}$ with input $x^{(i)}_{t+1}=0$, followed by running $\Out$ to generate a set of $n$ messages. It then runs $\estimate$ on the collection of $n$ messages, and returns $0$ if the answer is smaller than $n/2$, and $1$ otherwise. As before, $\rerandomize$ on a ciphertext $c=(t,\{s^{(i)}_t\}_{i=1}^n)$ (for $t<T$) will simulate, for each $i$, running the state transformations $\Statefn$ on  $s^{(i)}_{t}$ with input $x^{(i)}_{t+1}=0$ to get an updated state $s^{(i)}_{t+1}$. The updated ciphertext $c'$ is then set to $(t+1, \{s^{(i)}_{t+1}\}_{i=1}^n)$. 

As before, $\dec(\enc(b, k_{pub}), k_{priv})$ comprises an exact simulation of running the \countnonzero\ algorithm on $n$ clients, where the input streams  satisfy $\one(\x^{(i)})=b$, followed by rounding the final estimate. Since the correct answer to \countnonzero\ on these inputs is $bn$ and the estimation algorithm has error $n/4$ with high probability, we conclude that the encryption scheme satisfies $\dec(\enc(b, k_{pub}), k_{priv}) = b$ except with negligible probability.

Now \cref{lem:tv_small} implies that if the streaming algorithm was information-theoretically $\eps$-locally pan-private, then the $TV$ distance between the distributions $\enc(0,k_{pub})$ and $\enc(1,k_{pub})$ would be $O(\eps/\sqrt{T})$. \citet{MironovPRV09} show that computational $(\eps,\delta)$-indistinguishability (\cref{def:comp_dp}) is equivalent the existence of distributions $D_0$ and $D_1$  such that $D_0$ and $D_1$ are $O(\eps/\sqrt{T})$-indistinguishable, and  $D_b$ and $\enc(b, k_{pub})$ are computationally indistinguishable. 
It follows that $\enc(0, k_{pub})$ are $\enc(1, k_{pub})$ are computationally $(0,O(\eps/\sqrt{T} + T\cdot \negl(\lambda)))$-indistinguishable. 

The rerandomization, as before, is essentially simulating one step in the decoding. It is therefore immediate that for any valid ciphertext, $\dec(\rerandomize(c, k_{pub}), k_{priv})$ and $\dec(c, k_{priv})$ are running exactly the same simulation and thus are identically distributed. Finally, the security of the rerandomization follows once again from the hybrid argument: if we could distinguish iterated rerandomizations of $0$ from those of $1$, we would get a distinguisher that can learn the $b$ from the local state at some later time step $t$.

Finally, we argue computational efficiency. The $\enc$ and $\dec$ algorithms need to simulate applying $\Statefn^\tau$ on input $0^\tau$, for a suitable $\tau$, at most $2n$ times each. For a state space of size $S$, the transformation is defined by a stochastic matrix $M$ of size $2^{S} \times 2^{S}$. Applying this Markov kernel defined by $M$ $\tau$ times is equivalent to applying $M^\tau$. Since computing $M^\tau$ can be done in time $poly(2^S)\cdot \log \tau$ by repeated squaring, and thus we get the claimed run time for the scheme.
\end{proof}

\section{Conclusions}
\label{sec:conclusions}
In this work, motivated by privacy concerns on shared devices, we introduce the notion of local pan-privacy. We show that while information-theoretic local pan-privacy may be too strong a requirement for basic telemetry tasks, computational versions of this definition can be achieved without sacrificing on utility. We present algorithms for the fundamental tasks of counting the number of devices where a sensitive event occurs, as well as histograms of event counts, both in the trusted server and the two-server models. Our algorithms use public-key encryption schemes, and we show that such schemes are necessary to achieve computational local pan-privacy.

Our work raises many natural question. Our lower bound in \cref{thm:lower_bound} relies on instances with at most a single $1$, and shows a $\sqrt{T}$ gap in the error. We conjecture that this can be strengthened to an $\Omega(T)$ gap when one allows the instances to have arbitrarily many $1$s. While we have given algorithms for the most common telemetry tasks, other telemetry tasks may raise additional challenges. The validity proofs in the more general setting (when an adversary can collude with one of the servers) in our approach need to be non-interactive. We leave open the question of designing efficient zero-knowledge proofs for other predicates, that are compatible with local pan-privacy.

%% file: Sections/averages.tex
We recall the definition of additively homomorphic encryption scheme. Several popular public-key encryption schemes or standard variants of those, including Paillier, RSA and El Gamal, and lattice-based schemes, satisfy these properties.
\begin{definition}
\label{def:phe}
    A public key encryption scheme as defined in \cref{def:pke} is partially homomorphic over a group $(G,*)$ if there is a p.p.t. algorithm $\penc(\cdot,\cdot, k_{pub})$ with the property that for any pair of valid ciphertexts $c_1,c_2$, $\dec(\penc(c_1, c_2, k_{pub})) = \dec(c_1, k_{priv}) +  \dec(c_2, k_{priv})$; and a p.p.t. algorithm $\Phi_*$ that takes a group element and a ciphertext, and outputs another ciphertext with the property that for any $c$ and any $\alpha \in G$, $\dec(\Phi_{*}(\alpha, c, k_{pub}), k_{priv}) = \alpha * \dec(c, k_{priv})$.
\end{definition}
We will omit $k_{pub}$ as an argument from $\penc$ and $\Phi_{*}$ in the rest of the section and write $\penc(c_1,c_2,\ldots,c_k)$ to mean $\penc(c_1,\penc(c_2,\ldots,\penc(c_{k-1},c_k)\ldots))$.

Given an additively homomorphic encryption scheme, we will build on \cref{alg:histograms} to design a locally pan-private algorithm for means. Recall that the histogram algorithm already gives us encryptions of indicators of $|x|_1=j$ for each $j \in \{0,1,\ldots,k\}$. We will only need to redefine the \textbf{Send to Server} subroutine.
Building on these, the client computes an encryption of the number of occurrences by multiplying each encrypted bit by the number of occurrences and summing together. Since only one of the coordinates is an encryption of 1, the sum is the number of occurrences. The client can then privatize by adding discrete Gaussian noise~\cite{CanonneKS22, KairouzLS21} to the encrypted sum before sending to the server. 

\begin{algorithm}[H]
\caption{Averaging, Client Algorithm}\label{alg:average}
    \begin{algorithmic}[1]
    \Require $k, T, \enc, \penc, \Phi_{*}, \x, \sigma^2$

    \State \underline{\textbf{Send to Server}}
    \For{i = 0,1,\ldots,k}
        \State $s_i=\Phi_*(i, c_i)$
    \EndFor
    \State $r\sim\mathcal{N}_{\mathbb{Z}}(0, k\sigma^2)$
    \State $s_{k+1}=\enc(r)$

    \State \textbf{Return} $\penc(s_0, \cdots, s_{k+1})$
    \end{algorithmic}
\end{algorithm}

Upon receiving the reports from each client, the server simply decrypts and takes the average of the noisy reports. For an appropriate choice of $\sigma$ is easy to show the following.

\begin{theorem}\label{thm:ub_average}
       There is a computationally $0$-locally pan-private for estimating the sum, which is $(\eps_0,\delta_0)$-local DP with respect to the server. It estimates the sum with expected error $O\left(k\sqrt{n\log \tfrac 1 \delta_0}/\eps_0\right)$. Moreover, for any $(\eps,\delta) \in (0,1)$, there is an $\sigma$ such that the mechanism is $(\eps,\delta)$-aggregator DP, and has expected error $O(k\sqrt{\log\tfrac 1 \delta}/\eps)$.
\end{theorem}

%% file: Sections/appendix.tex
\input{Sections/binomial}

%% file: Sections/binomial.tex
\section{On the TV distance between Binomials}
\label{app:binomials_tv}
In this section, we will prove the following inequality.
\begin{theorem}
    Let $T \geq 2$ and $p \in (0,\frac 1 2)$. Then
    \begin{align*}
        TV(Bin(T, p), Bin(T-1, p) + Bern(1-p)) &= 2(1-2p)\frac{\ceil{Tp}}{Tp} \Pr[Bin(T,p) = \ceil{Tp}]\\
        &\leq (1-2p)/\sqrt{4p(1-p)T})
    \end{align*}
\end{theorem}
\begin{proof}
Note that $Bin(T,p) =  Bin(T-1,p) + Bern(p)$. Since we can couple $Bern(p)$ and $Bern(1-p)$ so that they agree with probability $2p$, and differ by $1$ otherwise, it follows that
\begin{align*}
    TV(Bin(T, p), Bin(T-1, p) + Bern(1-p)) &= (1-2p) \cdot TV(Bin(T-1, p), Bin(T-1, p) + 1).
\end{align*}
We now write
\begin{align*}
    2 TV(Bin(T-1, p), Bin(T-1, p) + 1) &= \sum_{m=0}^{T} |\Pr[Bin(T-1,p)=m] - \Pr[Bin(T-1,p) = m-1]|\\
    &= \sum_{m=0}^{T} |{T-1 \choose m} p^{m} (1-p)^{T-1-m} - {T-1 \choose m-1} p^{m-1} (1-p)^{T-m}|\\
    &= \sum_{m=0}^{T} \frac{1}{Tp(1-p)} \cdot {T \choose m} p^m (1-p)^{T-m} | p(T-m) - (1-p)m|\\
    &= \sum_{m=0}^{T} \frac{1}{Tp(1-p)} \cdot {T \choose m} p^m (1-p)^{T-m} | pT - m|\\
\end{align*}
Letting $f(m)$ denote $pT-m$, we can thus write
\begin{align*}
     2 TV(Bin(T-1, p), Bin(T-1, p) + 1) &= \frac{1}{Tp(1-p)} \E_{X \sim Bin(T, p)} [|f(X)|].
\end{align*}
The value $|f(X)|$ is the absolute deviation of a binomial random variable $Bin(T,p)$ from its expectation. An exact formula for the mean absolute deviation of a binomial was given by ~\citet{DeMoivre1730} (see ~\citet{DiaconisZ91} for a historical perspective on this formula). De Moivre showed that
\begin{align*}
    \E_{X \sim Bin(T, p)} [|X - Tp|] &= 2\ceil{Tp} (1-p) \Pr[Bin(T,p) = \ceil{Tp}]\\
    &= 2 {T \choose \ceil{Tp}} \ceil{Tp}  p^{\ceil{Tp}}(1-p)^{T-\ceil{Tp}+1}. 
\end{align*}
Plugging this bound gives the claimed equality. To prove the upper bound, we use Jensen's inequality:
     \begin{align*}
    \sqrt{\E_{X \sim Bin(T, p)} [|f(X)|]}&\leq  \sqrt{\E_{X \sim Bin(T, p)} [f(X)^2]}.\label{eq:mad_variance}
\end{align*}
Now observe that 

\begin{align*}
\E_{X \sim Bin(T, p)} [f(X)^2] & = \E_{X \sim Bin(T, p)} [(X-pT))^2]\\
&= p(1-p)T.
\end{align*}
It follows that
\begin{align*}
 2 TV(Bin(T-1, p), Bin(T-1, p) + 1) &\leq \sqrt{\frac{1}{p(1-p)T}},
\end{align*}
so that
\begin{align*}
TV(Bin(T, p), Bin(T-1, p) + Bern(1-p)) &\leq (1-2p) \sqrt{\frac{1}{4p(1-p)T}}.
\end{align*}
\citet{BerendK13} show that the bound from Jensen's inequality above (which is the only inequality in this proof) is tight up to a $\sqrt{2}$ factor for all $p \in (\frac 1 T, 1- \frac 1 T)$. When $p < \frac 1 T$, the mean absolute deviation is bounded by $2Tp$ and this is tight up to a factor of $e$. A symmetric bound holds for $p>1-\frac 1 T$.
\end{proof}